\documentclass{article}
\usepackage{spconf,amsmath,graphicx}

\usepackage{calc}
\usepackage{here}
\usepackage{hyperref}
\usepackage{xurl}
\usepackage{latexsym}
\usepackage{amsfonts}
\usepackage{amsmath,amsthm}   
\usepackage{amssymb}
\usepackage{bm}
\usepackage{capitalgreekitalic}
\usepackage{balance}

\DeclareGraphicsExtensions{.pdf,.jpeg,.png,.jpg,.emf,.eps}

\usepackage{geometry}
\usepackage[font=small,labelfont=bf]{caption}
\usepackage[font=small,labelfont=bf]{subcaption}

\usepackage{algpseudocode}
\usepackage{algorithmicx}
\usepackage[ruled,noline,linesnumbered]{algorithm2e}

\newtheorem{theorem}{Theorem}

\newtheorem{proposition}{Proposition}

\def\J{{\bf J}}
\def\R{{\bf R}}

\def\Q{{\bf Q}}
\def\I{{\bf I}}
\def\A{{\bf A}}
\def\E{{\bf E}}
\def\U{{\bf U}}
\def\V{{\bf V}}
\def\F{{\bf F}}

\def\G{{\bf G}}
\def\H{{\bf H}}
\def\U{{\bf U}}
\def\B{{\bf B}}
\def\C{{\bf C}}
\def\X{{\bf X}}
\def\x{{\bf x}}

\def\Thetab{\bm{\Theta}}
\def\Sigmab{\bm{\Sigma}}
\def\Lambdab{\bm{\Lambda}}

\def\mub{\bm{\mu}}

\def\det{\operatorname{det}}
\def\diag{\operatorname{diag}}

\newcommand{\CU}{\mathcal{U}}

\title{Riemannian optimization on the manifold of unitary and symmetric matrices with application to BD-RIS-assisted systems}
\name{Ignacio~Santamaria$^1$, Mohammad Soleymani$^2$, Eduard Jorswieck$^3$, Jes{\'u}s Guti{\'e}rrez$^4$, Carlos Beltr\'an$^5$ \thanks{This work is partly supported by the European Commission’s Horizon Europe, Smart Networks and Services Joint Undertaking, research and innovation program under grant agreement 101139282, 6G-SENSES project. The work of I. Santamaria was also partly supported under grant PID2022-137099NB-C43 (MADDIE) funded by MICIU/AEI /10.13039/501100011033 and FEDER, UE. }}
\address{\normalsize$^1$Department of Communications Engineering, Universidad de Cantabria, 39005 Santander, Spain\\ \normalsize
$^2$Signal and System Theory Group, Universit{\"a}t  Paderborn, 33098 Paderborn, Germany\\ \normalsize
$^3$Institute for Communications Technology, Technische Universit{\"a}t Braunschweig, 38106 Braunschweig,
Germany\\ \normalsize
$^4$IHP - Leibniz-Institut
f{\"u}r Innovative Mikroelektronik, 15236 Frankfurt (Oder), Germany\\ \normalsize
$^5$Departamento de Matem\'aticas, Estad\'istica y Computaci\'on, Universidad de Cantabria, 39005 Santander, Spain
}
\begin{document}
\ninept
\maketitle
\begin{abstract}
In this paper, we rigorously characterize for the first time the manifold of unitary and symmetric matrices, deriving its tangent space and its geodesics. The resulting parameterization of the geodesics (through a real and symmetric matrix) allows us to derive a new Riemannian manifold optimization (MO) algorithm whose most remarkable feature is that it does not need to set any adaptation parameter. We apply the proposed MO algorithm to maximize the achievable rate in a multiple-input multiple-output (MIMO) system assisted by a beyond-diagonal reconfigurable intelligent surface (BD-RIS), illustrating the method's performance through simulations. The MO algorithm achieves a significant reduction in computational cost compared to previous alternatives based on Takagi decomposition, while retaining global convergence to a stationary point of the cost function.    
\end{abstract}

\section{Introduction}
\label{sec:intro}
Manifold optimization methods have become an essential tool in modern signal processing, as many design problems involve constraints that naturally define smooth Riemannian manifolds \cite{AbsilBook,boumal2023intromanifolds,edelman1998geometry}. A particularly relevant manifold arises when dealing with unitary matrices ($\U\U^H=\I_n$). The unitary manifold, denoted as $\CU$, has been extensively studied, and its geometry is well understood, allowing the development of effective Riemannian optimization algorithms with applications in communications \cite{SunTSP24}, blind source separation \cite{AbrudanTSP08}, and machine learning \cite{ProjUNN}. 

In certain applications, however, the optimization matrix must be both unitary and symmetric ($\U = \U^T$), belonging to the manifold denoted in this paper as $\CU_s$. Despite its relevance, $\CU_s$ has received little attention in the literature, and its main elements such as the tangent space and geodesics remain unexplored. The main goal of this paper is to rigorously study the geometry of $\CU_s$ and leverage these results to design manifold optimization (MO) algorithms that operate directly on this manifold.

As an application, we consider the rate maximization of a multiple-input multiple-output (MIMO) link assisted by a beyond-diagonal reconfigurable intelligent surface (BD-RIS). In this context, the BD-RIS scattering matrix, $\Thetab$, is a unitary (passive lossless) and symmetric (due to reciprocity) matrix \cite{ClerckxTWC22a, ClerckxTWC22b}, making it a natural candidate for optimization over $\CU_s$. Previous approaches have circumvented the lack of characterization of $\CU_s$ by discarding the symmetry constraint and optimizing over $\CU$, as in \cite{EmilEuCap2025}; or by relaxing the problem and projecting the relaxed solution onto the set $\CU_s$, as in \cite{MaoCL2024}. These alternatives are clearly suboptimal and provide only an approximate solution. Other recent works \cite{SantamariaSPAWC24, Xia25} exploit the so-called Takagi factorization \cite{Takagi, Autonne}, which states that any unitary and symmetric matrix $\Thetab \in \CU_s$ can be decomposed as $\Thetab = \Q\Q^T$, where $\Q \in \CU$, but the computational cost of these alternatives is high. 

The main contributions of this work are:
\begin{itemize}
\item We present a study of the manifold $\CU_s$, characterizing its tangent space, its retraction mapping, and its geodesics.

\item Taking advantage of the parameterization of geodesics in $\CU_s$, we present an MO algorithm that does not require adaptation parameters.

\item As an application of the MO algorithm, we consider the rate maximization problem in a BD-RIS-assisted MIMO system, comparing its performance with previous algorithms proposed in the literature.

\end{itemize}

\indent \textit{Notation}: The symbols for scalars, vectors, matrices, and sets are, respectively, $x$, $\x$, $\X$, and $\mathcal{X}$. $\A^T$, $\A^*$, $\A^H$, $\A^{-1}$, $\A^{1/2}$, $\det(\A)$ are, respectively, transpose, conjugate, Hermitian, inverse, square root, and determinant of matrix $\A$. $\mathrm{Imag}(a)$ denotes the imaginary part of $a$, and $j$ denotes the imaginary unit. $\I_n$ denotes the identity matrix of size $n$. ${\cal CN}({\bf 0}, {\bf R})$ is the proper complex Gaussian distribution with zero mean and covariance matrix ${\bf R}$.

\section{Manifold optimization on $\CU_s$}
\subsection{The manifold $\CU$}
It is well known that $\CU=\{\U \in \mathbb{C}^{n\times n}: \U\U^H=\I_n\}$ is a real Riemannian manifold of dimension $n^2$ \cite{AbsilBook,boumal2023intromanifolds,edelman1998geometry}, with tangent space $T_{\U} \CU=\{\B\in \mathbb{C}^{n\times n}: \U^H\B+\B^H\U=0\}$. Alternatively, the tangent space at $\U$ can be parametrized as $\B = \U {\bf S}$, where ${\bf S}$ is an $n \times n$ skew-Hermitian matrix.
The Riemannian exponential that maps a point in the tangent space to the manifold is defined in terms of the matrix exponential as
\[
\begin{matrix}
	\exp_{\CU,\U}:&T_{\U}\CU&\to&\CU\\
	&\B&\to&\U e^{\U^H\B}
\end{matrix}
\]
which is a diffeomorphism if restricted to tangent matrices $\B$ of sufficiently small Frobenius norm \cite{HallLieGroups}.

\subsection{The manifold $\CU_s$}
Our objective is to study the set $ \CU_s= \{\U \in \CU: \U=\U^T\}$ and then derive a Riemannian optimization algorithm. Firstly, we recall Takagi's decomposition of a complex and symmetric matrix, which plays an important role in the characterization of $\CU_s$. 
\begin{theorem}[Takagi factorization \cite{Takagi}]
Let $\A = \A^T$ be an $n \times n$ complex symmetric matrix. Then, there exist an $n \times n$ unitary matrix $\Q \in \CU$ and an $n\times n$ diagonal matrix $\bm{ \Sigma} = \diag(\sigma_1,\ldots,\sigma_n)$ with $\sigma_1 \geq \sigma_2 \geq \ldots \geq \sigma_n \geq 0$ such that $\A = \Q \bm{ \Sigma} \Q^T$.
\end{theorem}

From the singular value decomposition (SVD) of a symmetric matrix $\A = \F \Sigmab \G^H$, its Takagi decomposition can be obtained as $\A = \Q \Sigmab \Q^T$ with $\Q = \F (\F^H\G^*)^{1/2}$. Note that if all singular values of $\A$ are distinct,  $\F^H\G^*$ is a diagonal matrix and the factor $\Q$ is obtained as described in \cite[Remark 2]{SantamariaSPLetters2023}. We have the following propositions:
\begin{proposition} [{\bf tangent space and geodesics}]
\label{prop:tangent}
	$\CU_s$ is a real Riemannian manifold of dimension $n(n+1)/2$ and its tangent space at the point $\U \in\CU_s$ is
	\begin{align*}
	T_{\U}\CU_s=&\{\B\in  \mathbb{C}^{n\times n}: \U^H\B+\B^H\U={\bf 0},\;\B=\B^T\}\\
    =& \{j\Q\R\Q^T:\R\in \mathbb{R}^{n\times n}, \R=\R^T \},
	\end{align*}
    where $\U=\Q\Q^T$ is a Takagi decomposition of $\U$. Therefore, the tangent space at $\U =  \Q\Q^T$ is composed of all matrices $\B = j \Q\R\Q^T$ with $\R$ an  $n\times n$ real and symmetric matrix. In addition, the Riemannian exponential is the restriction of $\exp_{\CU,\U}$ to $T_{\U}\CU_s$. Therefore, the geodesic in $\CU_s$ starting at $\U =\Q\Q^T$ and with direction $\B = j\Q\R\Q^T \in  T_{\U}\CU_s$ is
    \begin{equation*}
        \U(\mu) = \U e^{\U^H\B \mu} = \U e^{j \Q^* \R \Q^T \mu},
    \end{equation*}
    with $\R$ real and symmetric, and $\mu \geq 0$ a step size that parameterizes the geodesic.
\end{proposition}

\begin{proof}
    See Appendix.
\end{proof}

The above result can be rewritten as follows. The matrix $j\R$ (purely imaginary skew-Hermitian), admits the following eigendecomposition $j\R = \V_R \Sigmab_R \V_R^T$, where $\V_R$ is a (real) orthogonal matrix and $\Sigmab_R = \diag(j\theta_1,\ldots, j\theta_n)$ is a diagonal matrix with the imaginary eigenvalues. Applying standard results of the exponentiation of diagonalizable matrices, we see that the geodesic starting at $\U = \Q\Q^T$ and pointing in the direction $j\Q\R\Q^T$ can be written as
\begin{eqnarray}
     \U(\mu) =  \U e^{j \Q^* \R \Q^T \mu} &=& \Q\Q^T \Q^* \V_R \Lambdab_R^\mu \V_R^T \Q^T  \nonumber \\
     &\stackrel{(a)}{=}& \Q \V_R \Lambdab_R^\mu \V_R^T \Q^T \nonumber \\
     & \stackrel{(b)}{=}& \Q_R  \Lambdab_R^\mu \Q_R^T \label{eq:geodesic2}
\end{eqnarray}
where in $(a)$ we have used $\Q^T \Q^*  = \I_n$, in $(b)$ the unitary matrix $\Q \V_R$ is denoted as $\Q_R$, and $\Lambdab_R^\mu$ is the following diagonal matrix   
\begin{equation*}
\Lambdab_R^\mu = \exp(\Sigmab_R)^\mu = \diag \left(e^{j\theta_1 \mu},\ldots, e^{j\theta_n \mu} \right),
\label{eq:RISmatrix}
\end{equation*}
which is essentially a diagonal RIS. The formulation of the geodesic in \eqref{eq:geodesic2} is the main result of this work, as it establishes an interesting connection between the optimization of a unitary and symmetric BD-RIS and that of a diagonal RIS. In words, geodesics in $\CU_s$ are traversed by multiplying the phases of a diagonal RIS, $\Lambdab_R = \exp(\Sigmab_R)$, by a step size $\mu$. The direction of the geodesic, which will depend on the cost function to be optimized, is determined by the eigenvectors and eigenvalues of $\R$. 

\begin{proposition} [{\bf Retraction to} $\CU_s$ ] \label{prop:retraction} The mapping 
\[
	\begin{matrix}
		\Pi:&\{ \A\in \mathbb{C}^{n\times n}: \A=\A^T\}&\to&\CU_s\\
		&\A&\to& \Q\Q^T,
	\end{matrix}
\]
where $\A=\Q \Sigmab \Q^T$ is a Takagi factorization of $\A$, sends $\A$ to the unitary and symmetric matrix $ \U = \Q\Q^T \in \CU_s$ closest to $\A$ in Frobenius norm (to one of them, if there are several).
\end{proposition}

\begin{proof}
    Given $\A = \F  \Sigmab  \G^H$ it is known that $\Pi(\A) = \F\G^H$ is the unitary matrix closest to $\A$. If $\A=\Q \Sigmab \Q^T$ is symmetric, $\Pi(\A) = \Q\Q^T$ is hence the unitary and symmetric matrix closest to $\A$ in $\CU$ and consequently in $\CU_s \subset \CU$.
\end{proof}

\begin{proposition} [{\bf Projection to the tangent space}] \label{prop:projection} Given any matrix $\J \in \mathbb{C}^{n\times n}$, the orthogonal projection of $\J$ onto $T_{\U}\CU_s$ is 
\begin{equation}
		\pi_{T_{\U}\CU_s} \J = j\,\Q \underbrace{\mathrm{Imag}\left(\frac{\Q^{H}(\J+\J^T)\Q^{*}}{2}\right)}_{\R}\Q^T,
        \label{eq:proj}
\end{equation}	
	with $\U = \Q\Q^T$ a Takagi decomposition of $\U \in \CU_s$.
\end{proposition}

\begin{proof} According to Proposition \ref{prop:tangent} the tangent space at $\U = \Q\Q^T$ can be parameterized as $\B = j\Q\R\Q^T$ with $\R$ real and symmetric. Solving the problem 
\[\min_{\R \in \mathbb{R}^{n \times n}, \R = \R^T} \| {\bf J}- j\Q\R\Q^T \|_F^2
\] 
we get \eqref{eq:proj}.
\end{proof}

\subsection{Manifold optimization algorithm}
Suppose we want to maximize a continuous and differentiable function, $f(\U)$, with $\U \in \CU_s$. To develop an MO algorithm on $\CU_s$, it remains to be discussed how to select the step size, $\mu$, along the geodesic in \eqref{eq:geodesic2}. The standard approach to find a suitable $\mu$ is to apply a line search procedure that ensures $f(\U_{k+1})>f(\U_k)$, as described in \cite[Ch. 4]{AbsilBook}. Nevertheless, the parameterization of the geodesic in \eqref{eq:geodesic2} allows us to conceive other methods that increase the search space by optimizing a vector $\mub = (\mu_1,\ldots,\mu_n)$ instead of optimizing a single adaptation step $\mu$ common to all phases $\theta_m$. More specifically, the new phases, denoted as $\theta_m' = \theta_m \mu_m$, $m=1,\ldots, n$, can be optimized one by one alternately, sometimes with closed-form solutions for certain cost functions. The proposed algorithm with this alternating optimization procedure for the selection of the new phases is summarized in Algorithm  \ref{alg:MOopt}\footnote{Matlab code can be downloaded from \url{https://github.com/IgnacioSantamaria/Code-MaxCap-BD-RIS-MIMO}.}. 

\begin{proposition} [{\bf Convergence}] \label{prop:convergence} Suppose that i) $f(\U)$ is a bounded from above, continuous, and differentiable function that we wish to maximize; and ii) the vector of step sizes $\mub = (\mu_1,\ldots,\mu_n)$ is selected to satisfy $f(\U_{k}) > f(\U_{k-1})$. Then, under mild conditions (see \cite[Ch. 4]{AbsilBook}), the sequence of iterates $\U_0,\U_1, \ldots$ provided by Algorithm \ref{alg:MOopt} converges to a stationary point of $f(\U)$.  
\end{proposition}

\begin{proof} 
The proof of the global convergence of monotone line search Riemannian algorithms on manifolds can be found in \cite[Ch. 4]{AbsilBook}.
\end{proof}

\begin{algorithm}[!t]
\small
\DontPrintSemicolon
\SetAlgoVlined
\KwIn{Initial $\U_0 = \Q_0\Q_0^T \in \CU_s$, Cost function: $F_0 = f(\U_0)$, convergence threshold: $\epsilon$}
\KwOut{Final $\U = \Q\Q^T \in \CU_s$}
\For{$k = 1,\ldots$}{
Calculate gradient $\J =\nabla f(\U)$ at $\U_{k-1} = \Q_{k-1}\Q_{k-1}^T$ \;
Find $\R$ as in \eqref{eq:proj}\;
Perform eigendecomposition $j\R = \V_R \Sigmab_R \V_R^T$\;
$\Lambdab_R = \exp(\Sigmab_R) = \diag \left(e^{j\theta_1},\ldots, e^{j\theta_n} \right)$\;
$\Q_{R} = \Q_{k-1} \V_R$ \;
\For{$m=1,\ldots,n$}{
$\theta_m' = {\rm arg}\max_{\theta_m} f(\Q_{R} \Lambdab_R(\theta_m) \Q_{R}^T)$}

Obtain the new Takagi factor as $\Q_{k} = \Q_{k-1} \V_R \Lambdab_R^{1/2}$\;
Obtain $\U_{k} = \Q_{k}\Q_{k}^T$ and $F_k = f(\U_{k})$\;
Check convergence: $|F_k-F_{k-1}| < \epsilon$ 
}
\caption{Manifold optimization algorithm in $\CU_s$}
\label{alg:MOopt}
\end{algorithm}

\section{Application}
\label{sec:rateopt}
\subsection{Rate optimization of BD-RIS-assisted MIMO links}
\label{sec:problem}
As an application of Algorithm \ref{alg:MOopt}, we consider the design of a BD-RIS in a MIMO channel to maximize the achievable rate. The MIMO system, with $N_t$ transmit antennas and $N_r$ receive antennas, is assisted by a fully connected BD-RIS with $M$ elements. The equivalent MIMO channel is
\begin{equation}
\H_{eq} (\Thetab) = \H_d + \F \Thetab \G^H,
\label{eq:MIMOchanneleq}
\end{equation} 
where $\G \in \mathbb{C}^{N_t \times M}$ is the channel from the Tx to the BD-RIS, $\F \in \mathbb{C}^{N_r \times M}$ is the channel from the BD-RIS to the Rx, $\H_d \in \mathbb{C}^{N_r \times N_t}$ is the MIMO direct link, and $\bm{\Theta} \in \CU_s$ is the $M \times M $ BD-RIS matrix.

The Tx sends proper Gaussian signals with a fixed isotropic covariance matrix, $\x \sim \mathcal{CN}({\bf 0}, P \I)$, and the received signal is contaminated by additive white Gaussian noise, ${\bf n} \sim \mathcal{CN}({\bf 0}, \sigma^2 \I)$. Then, the BD-RIS that maximizes the achievable rate solves the problem
\begin{equation*}
\label{eq:MaxCapMIMO}
\,\max_{\Thetab \in \CU_s}\,\, f(\Thetab) = \max_{\Thetab \in \CU_s}\log \det \left( \I_{N_R} + \frac{P}{\sigma^2}\H_{eq}(\Thetab)\H_{eq}(\Thetab)^H \right), 
\end{equation*}
where $\H_{eq}(\Thetab)$ is given by \eqref{eq:MIMOchanneleq}. Let us define for simplicity of notation $\E(\Thetab) = \I_{N_R} + \frac{P}{\sigma^2}\H_{eq}(\Thetab)\H_{eq}(\Thetab)^H$. The unconstrained gradient of the cost function is \cite{PalomarTIT06}
\begin{equation}
    \J = \nabla f(\Thetab) = \frac{P}{\sigma^2} \F^H \E(\Thetab)^{-1} \H_{eq}(\Thetab) \G. 
    \label{eq:unconstrained}
\end{equation}

The projection of $\nabla f(\Thetab)$ onto the tangent space gives us the symmetric and real matrix $\R$ in \eqref{eq:proj}. From the eigendecomposition of $\R$, the new BD-RIS can be parameterized as $\Thetab = \Q_R \diag \left(e^{j\theta_1'},\ldots, e^{j\theta_M'} \right) \Q_R^T$. We can now substitute this BD-RIS parameterization into the rate expression and apply an alternating optimization method to independently optimize each of the phases $\theta_m'$. To solve this problem, we can use the method proposed in \cite[Sec. III.B]{ZhangCapacityJSAC2020}.

\subsection{Computational cost}
\label{sec:compcost}
To simplify the notation, we consider the maximization of the rate of an $N \times N$ MIMO system. For this particular case, the computation of the gradient in \eqref{eq:unconstrained}  has a complexity $\mathcal{O}(MN^2 + M^2N +N^3)$. The projection on the tangent space involves the multiplication of $M \times M$ matrices and is $\mathcal{O}(2M^3)$, and the eigendecomposition of $\R$ is $\mathcal{O}(M^3)$. The update of the Takagi factor to $\Q_k = \Q_{k-1} \V_R \Lambdab_R^{1/2}$ is also $\mathcal{O}(M^3)$. Finally, the computational cost per iteration grows as $\mathcal{O}(4M^3 + MN^2 + M^2N +N^3)$, which is dominated by the first term.

\section{Simulation results}
\label{sec:simulations}
In this section, numerical results are provided to validate the effectiveness of the proposed MO algorithm. We consider a $4 \times 4$ MIMO system with the Tx located at coordinates (0,0,1.5) [m] and the Rx at coordinates (50,0,1.5) [m]. A fully-connected BD-RIS with $M$ elements is located close to the Rx at coordinates (50,3,3) [m] to assist the Tx-Rx communication. The channels through the BD-RIS have a dominant line-of-sight path and are therefore modeled as Rician with factor $K=3$ and path loss exponent $\alpha = 2$. The direct channel is modeled as a Rayleigh channel with path loss exponent that we take as $\alpha = 3.75$ or $\alpha = 8$ to model scenarios where the direct link is obstructed. The rest of the scenario parameters are the same as in \cite{SantamariaSPAWC24}. The results shown are the average of 200 independent simulations (channel realizations) keeping the Tx, Rx and BD-RIS positions fixed. For the proposed algorithm, we used a convergence threshold of $\epsilon = 10^{-3}$. We compare the performance of the proposed algorithm (labeled as {\bf MO in $\CU_s$}) with the following methods:
\begin{itemize}
     \item {\bf MM+MO in $\CU$} (\cite{SantamariaSPAWC24}): This method applies a minorize-maximization (MM) algorithm and exploits the Takagi factorization of the BD-RIS. 
    \item {\bf MO in $\CU$ + Proj}: it applies an MO algorithm on $\CU$ to find a unitary but non-symmetric BD-RIS, $\Thetab_u$, that maximizes the rate. The final BD-RIS is obtained as the retraction (see Proposition \ref{prop:retraction}) $\Pi(\Thetab_u + \Thetab_u^T)$.
    \item {\bf Low Cost} (\cite{MaoCL2024}): The BD-RIS is obtained as $\Pi(\F^H\H_d\G + (\F^H\H_d\G)^T)$. It assumes that the direct link is not blocked.
\end{itemize}

\begin{figure}[htb]
    \centering
\includegraphics[width=.5\textwidth]{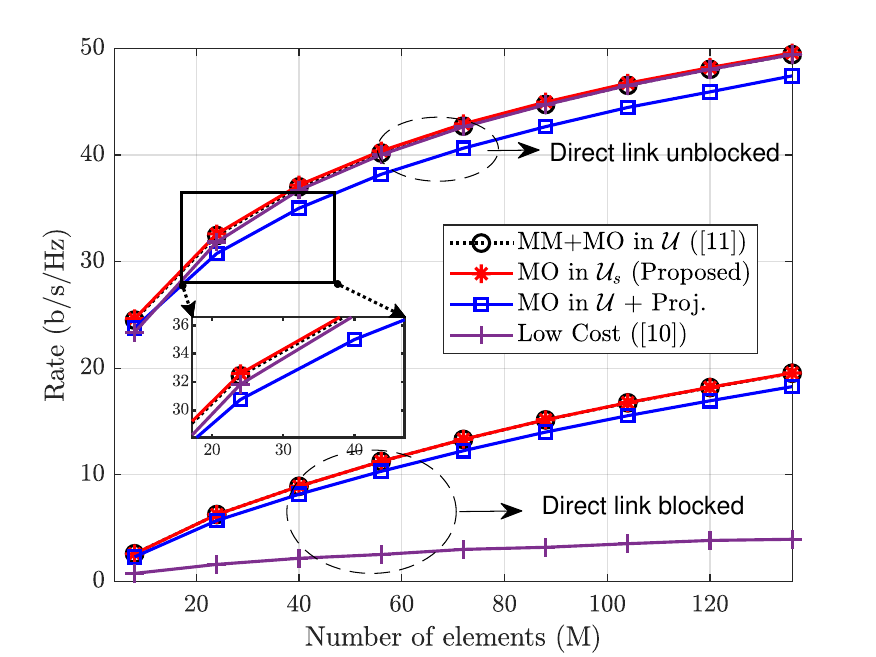}
     \caption{Rate vs. number of BD-RIS elements achieved with different algorithms.}
	\label{fig:RatevsM}
\end{figure}

Fig. \ref{fig:RatevsM} shows the rate versus the number of BD-RIS elements, $M$, when the direct channel is blocked (lower curves) and when there is a direct channel (upper curves). The proposed MO algorithm provides practically identical results to \cite{SantamariaSPAWC24} in both cases, although, as we will see later, with much less computational effort. The results are also very similar to those provided by \cite{MaoCL2024} when the direct channel is unblocked, but when the direct channel is blocked this method is not applicable. Finally, the final retraction step after optimizing on $\CU$ causes a significant rate penalty.

Fig. \ref{fig:ConvergenceSameChannel4x4M64} a) illustrates the convergence of the proposed algorithm for a BD-RIS with $M=64$ elements starting from different random points. The largest rate increase occurs in the first iteration, and convergence is ensured in all realizations within less than 10 iterations. Fig. \ref{fig:ConvergenceSameChannel4x4M64} b) shows the computational complexity of the different algorithms under comparison measured through the logarithmic run time in secs. The proposed algorithm is orders of magnitude faster than the one described in \cite{SantamariaSPAWC24} and is even faster than the method {\bf MO in $\CU$ + Proj}. The closed-form solution of \cite{MaoCL2024} is the most computationally efficient, but it requires an unblocked direct channel.
\begin{figure}[ht]
    \centering
\includegraphics[width=.5\textwidth]{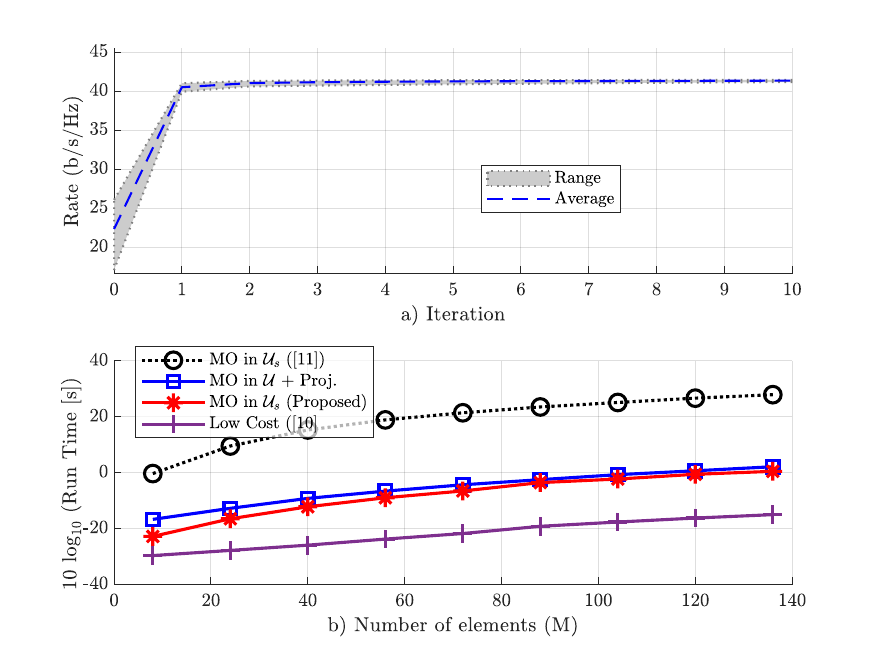}
     \caption{a) Convergence curves of the proposed algorithm starting from different random initializations, b) Run time (in log scale) vs. number of BD-RIS elements.}
	\label{fig:ConvergenceSameChannel4x4M64}
\end{figure}
\section{Conclusions}
The characterization of the tangent space and the geodesics of the manifold of unitary and symmetric matrices has allowed us to propose a new Riemannian optimization algorithm that does not require any adaptation parameter (e.g., step size). The application of the proposed MO algorithm to optimizing the rate of a BD-RIS-assisted MIMO system shows a significant reduction in computational cost with identical performance compared to previous alternatives based on Takagi decomposition. In future work, we will consider the application of the proposed MO algorithm to other problems, such as the minimization of the mean square error (MSE) \cite{LongTWC24} in BD-RIS assisted systems.

\section{Appendix: Proof of Proposition 1}
\label{AppendixA}
{\bf tangent space}. It is a simple exercise to verify that the following sets are equivalent vector spaces of dimension $n(n+1)/2$  
\begin{align*}
	T_U\CU_s =& \{\B\in  \mathbb{C}^{n\times n}: \U^H\B+\B^H\U={\bf 0},\;\B=\B^T\}\\
    =&\{\U \A: \A\in\mathbb{C}^{n\times n}, \A+\A^H= {\bf 0},\; \U\A\U^H=\A^T\} \\ 
    =& \{j\Q\R\Q^T:\R\in \mathbb{R}^{n\times n}, \R=\R^T \},	
\end{align*}
where $\U = \Q\Q^T \in \CU_s$ is a Takagi decomposition.\\

\noindent{\bf Geodesic}. The Riemannian exponential in $\CU$, $\phi(\B)= \U e^{\U^H\B}$, is a diffeomorphism (a bijection with differentiable inverse) when restricted to matrices $\|\B\|_F^2 < \epsilon$ \cite{HallLieGroups}. Therefore, its restriction to the set $\CU_s$ (i.e., when $\B \in  T_U\CU_s$) is a diffeomorphism with its image. It remains to prove that the set 
\[
\mathcal{X} = \{\phi(\B) : \B \in  T_U\CU_s, \|\B\|_F^2 < \epsilon \}
\]
contains an open set of $\B$ in $\CU_s$. We prove the double inclusion:
\begin{itemize}
    \item $\mathcal{X} \in \CU_s$: We have that $\phi(\B) \in \CU$ because it is a restriction of the exponential map in $\CU$. We have to prove that $\phi(\B)$ is symmetric.
\begin{eqnarray*}
    \phi(\B)^T = \left( \U e^{\U^H\B}\right)^T = \left( \U e^{\U^H\B \U^H \U}\right)^T \\
    = \left( \U \U^H e^{\B\U^H} \U\right)^T = \U e^{(\B\U^H)^T} = \phi(\B).
\end{eqnarray*}

\item The fact that for every $\U' \in \CU_s$ in a neighborhood of $\U$ there exists a $\B \in \mathcal{X}$ such that $\phi(\B) = \U'$ follows from the fact that the matrix exponential is locally invertible at ${\bf 0} $ since $\frac{d(e^{t\C})}{dt} |_{t=0} = \C$.
    
\end{itemize}

\bibliographystyle{IEEEbib}
\bibliography{refs}

\end{document}